\documentclass[11pt,a4paper,reqno]{article}
\usepackage[T2A]{fontenc}
\usepackage[utf8]{inputenc}
\usepackage[english]{babel}
\usepackage{fullpage}
\advance\textheight by 1cm

\title{The power of the Binary Value Principle}
\author{Yaroslav Alekseev\thanks{Steklov Institute of Mathematics at St.~Petersburg, and Technion. Supported by Lady Davys Fellowship.} 
\ and 
Edward A. Hirsch\thanks{Technion. Partially supported by the European Union’s Horizon
2020 research and innovation programme under grant agreement No 802020-ERC-HARMONIC.}}
\date{}

\usepackage{xifthen}
\usepackage{graphicx}
\usepackage{amsthm}
\usepackage{amsmath}
\usepackage{amssymb}
\usepackage{ dsfont }
\usepackage{latexsym}
\usepackage{xcolor}
\usepackage{euscript}
\usepackage{enumitem}
\usepackage[notransparent]{svg}
\definecolor{mycyan}{rgb}{0.0117647,0.533333,0.647059}

\usepackage[pdfstartview=FitH,pdfpagemode=None,
            backref,colorlinks=true,citecolor=mycyan,
            linkcolor=mycyan]{hyperref} \hypersetup{
              urlcolor=[rgb]{0.8, .2, .5}
}

\newtheorem{definition}{Definition}

\newtheorem*{note}{Note}

\newtheorem{theorem}{Theorem}[section]
\newtheorem{lemma}{Lemma}[section]

\newtheorem{corollary}{Corollary}[section]

\newcommand{\car}{{\rm CARRY}}
\newcommand{\cari}[1]{\ensuremath{\car_{#1}}}

\newcommand{\add}{{\rm ADD}}
\newcommand{\addi}[1]{\ensuremath{\add_{#1}}}
\newcommand{\addv}{\ensuremath{\overline \add}}
\newcommand{\prd}{{\rm PROD}}

\newcommand{\prdv}{\ensuremath{\overline \prd}}
\newcommand{\prdvp}{\ensuremath{\overline \prd_+}}
\newcommand{\iabs}{{\rm ABS}}
\newcommand{\iabsv}{\ensuremath{\overline \iabs}}
\newcommand{\val}{{\rm VAL}}
\newcommand{\sgn}{{\rm SIGN}}

\newcommand{\arit}[1]{\ensuremath{{\rm arit\!}\left(#1\right)}}

\newcommand{\vy}{\ensuremath{\overline y}}
\newcommand{\vz}{\ensuremath{\overline z}}
\newcommand{\ve}{\ensuremath{{\mathbf e}}}
\newcommand{\vm}{\ensuremath{\overline m}}

\newcommand{\vb}{\ensuremath{\overline b}}
\newcommand{\va}{\ensuremath{\overline a}}
\newcommand{\vx}{\ensuremath{\overline x}}
\newcommand{\vr}{\ensuremath{\overline r}}

\newcommand\poly{\hbox{\textup{poly}}}

\newcommand{\zbit}{{\rm BIT}}
\newcommand{\biti}[1]{\ensuremath{\zbit_{#1}}}
\newcommand{\bitv}{\ensuremath{\overline \zbit}}

\newcommand{\PC}{\ensuremath{\mathsf{PC}}}
\newcommand{\LS}{\ensuremath{\mathsf{LS}}}
\newcommand{\IPS}{\ensuremath{\mathsf{IPS}}}
\newcommand{\CPS}{\ensuremath{\mathsf{CPS}}}
\newcommand{\extpc}{\ensuremath{\mathsf{Ext}\textrm{-}\mathsf{PC}}}
\newcommand{\ebvp}{\ensuremath{\mathsf{eBVP}}}

\newcommand{\extps}{\ensuremath{\mathsf{Ext}\textrm{-}\mathsf{LS}_{+,*}}}
\newcommand{\extpsZ}{\ensuremath{\mathsf{Ext}\textrm{-}\mathsf{LS}_{+,*,\mathbb{Z}}}}
\newcommand{\extpcsurd}{\ensuremath{\mathsf{Ext}\textrm{-}\mathsf{PC}^\surd}}
\newcommand{\extpcsurdz}{\ensuremath{\mathsf{Ext}\textrm{-}\mathsf{PC}_{\mathbb{Z}}^\surd}}
\newcommand{\extpcsq}{\extpcsurd}
\newcommand{\extpcZ}{\ensuremath{\mathsf{Ext}\textrm{-}\mathsf{PC}_{\mathbb{Z}}}}
\newcommand{\extpcsqZ}{\ensuremath{\mathsf{Ext}\textrm{-}\mathsf{PC}_{\mathbb{Z}}^\surd}}
\newcommand{\extpcsqQ}{\ensuremath{\mathsf{Ext}\textrm{-}\mathsf{PC}_{\mathbb{Q}}^\surd}}
\newcommand{\pcsq}{\ensuremath{\mathsf{PC}^\surd}}

\begin{document}
\maketitle

\begin{abstract}
The (extended) Binary Value Principle ({\ebvp}, the equation $\sum_{i=1}^n x_i2^{i-1} = -k$ for $k>0$ and in the presence of $x^2_i=x_i$) has received a lot of attention recently,
several lower bounds have been proved for it \cite{AGHT19_new,alekseev2020lower,PT_bvp_journal}.
Also it has been shown \cite{AGHT19_new} that the probabilistically verifiable
Ideal Proof System ({\IPS}) \cite{GP14} together with {\ebvp} polynomially simulates a similar semialgebraic proof system.
In this paper we consider Polynomial Calculus with the algebraic version of Tseitin's extension rule ({\extpc}).
Contrary to {\IPS}, this is a Cook--Reckhow proof system.
We show that in this context {\ebvp} still allows to simulate similar semialgebraic
systems. We also prove that it allows to simulate the Square Root Rule \cite{GH03},
which is in sharp contrast with the result of \cite{alekseev2020lower}
that shows an exponential lower bound on the size of {\extpc} derivations of the Binary Value Principle from its square.
On the other hand, we demonstrate that {\ebvp} probably does not help in proving
exponential lower bounds for Boolean formulas: we show that
an {\extpc} (even with the Square Root Rule) derivation of any unsatisfiable Boolean formula in CNF from {\ebvp} must be of exponential size.
\end{abstract}

\tableofcontents

\section{Introduction}
Tseitin's extension rule \cite{Tse68} is a powerful concept that turns even very weak propositional proof systems into strong ones: it allows to introduce new variables for arbitrary formulas (it is enough to do this for the disjunction and the negation). In particular, it turns Resolution (a rather weak system for which superpolynomial lower bounds are known since \cite{Tse68}) into the powerful Extended Frege system \cite{CR79} (a strong system for which we do not even know good enough candidates for superpolynomial lower bounds). 

Surprisingly, in the context of algebraic proof systems an exponential lower bound for a system that uses Tseitin's rule was proved recently \cite{alekseev2020lower}.
This system, Extended Polynomial Calculus (or {\extpc}), combines the algebraic version of the extension rule (so that we can introduce new variables for polynomials) with the Polynomial Calculus ({\PC}) \cite{CEI96} system. While it has more power because it allows to talk about polynomials over any algebraically closed field (or, in the Boolean setting, even just over a ring, such as $\mathbb{Z}$), the exponential lower bound has been proved for a system of polynomial equations that does not correspond to any Boolean formula (in particular, a formula in conjunctive normal form, as in Resolution). This system, called ``the Binary Value Principle'', is the equation $\sum_{i=1}^nx_i2^{i-1}+1=0$ along with the ``Boolean axioms'' $x_i^2-x_i=0$ for every variable $x_i$. It has also been used for proving other exponential lower bounds \cite{AGHT19_new,PT_bvp_journal} and (as the Extended Binary Value Principle, {\ebvp}) for demonstrating a polynomial simulation of polynomial inequalities by polynomial equations \cite{AGHT19_new} for generalized proof systems that require polynomial identity testing for the verification (the algebraic system is the Ideal Proof System, {\IPS}, of \cite{GP14}).
Note that polynomial inequalities are considered to be much more powerful than polynomial equations: for example, no exponential size lower bound is known even for the simplest proof system {\LS} (motivated by the optimization procedure by L\'aszl\'o Lov\'asz and Alexander Schrijver, see \cite{Pud99} and \cite{GHP02}).

\emph{Our results.} In this paper, we consider three questions about {\extpc} and {\ebvp}, and prove three results:
\begin{enumerate}
\item How powerful is {\extpc}? 
We prove (Theorem~\ref{th:ext-pc-sim}) that together with {\ebvp} it polynomially simulates a similar system that uses inequalities (namely, {\extps}, which is {\LS} with extension variables, squares, and multiplication). This brings the result of \cite{AGHT19_new} down to conventional proof systems from proof systems that use polynomial identity testing for proof verification. It is interesting how far we can weaken the proof systems to keep such simulation (it becomes harder and harder when we weaken them to formulas or polynomials written as sums of monomials).
\item Grigoriev and Hirsch \cite{GH03} introduced the square root rule that allows to conclude $f=0$ from $f^2=0$. It would be needed for the implicational completeness of {\PC} in the non-Boolean case. It is not needed at all in the Boolean context, however, it could shorten the proofs. It is impossible to simulate it polynomially in {\extpc} (\cite{alekseev2020lower} proves an exponential bit-size-of-coefficients lower bound on derivations of $\sum x_i2^{i-1}+1=0$ from $(\sum x_i2^{i-1}+1)^2=0$) and {\PC} (\cite{PTT21} proves a linear degree lower bound on derivations of $\sum x_i+1=0$ from $(\sum x_i+1)^2=0$). We prove (Theorem~\ref{th:sq-root}) that in the case of {\extpc} derivations the square root rule can be polynomially simulated using {\ebvp}.
\item Is it possible to use lower bounds for {\ebvp} for proving lower bounds for formulas in conjunctive normal form? One could imagine deriving the translation of an unsatisfiable formula in conjunctive normal form (using the extension variables) from {\ebvp} and concluding a lower bound for a formula in CNF. We prove an exponential lower bound (Theorem~\ref{lower bound q CNF}) on the size of derivations of such formulas from {\ebvp}, showing an obstacle to this approach.
\end{enumerate}

\emph{Our methods.}
The divisibility method suggested in \cite{AGHT19_new,alekseev2020lower} allows to prove lower bounds on the size of algebraic proofs
by analysing the scalars appearing in them. The simplest application of this method substitutes the input variables 
by the binary representations of all possible integers, and shows that the constant in the final contradiction
in the proof over the integers divides all of them (if the system allows it).
In this paper we further develop this method: we prove lower bounds for the derivation of a translation of an unsatisfiable formula in CNF (and not just a contradiction), so there is no single constant at the end. We show an exponential lower bound over the integers
by counting the primes that divide  the multiplicative constants in the derivation of every clause and Boolean equation.
The lower bound for rationals follows using the translation of \cite{alekseev2020lower}.

In order to show polynomial simulations we use the general approach suggested in \cite{AGHT19_new}: 
to use bit arithmetic for proving facts about (semi)algebraic proofs. 
However, {\IPS} \cite{GP14} considered in that paper uses polynomial identity testing for proof verification,
and thus allows to switch between the circuit representations of polynomials at no cost.
Our setting is different: we need to simulate everything using the extension rule.
Therefore, in order to simulate inequalities
we derive gradually the facts that the values produced by bit arithmetic
equal the values of polynomials in the original proof, and that these values are nonnegative.
We also need to define the circuit representation, in particular for the extension variables,
in order to reason about {\extps} proofs.
A somewhat similar approach works for the simulation of the square root rule;
however, we need to derive that all individual bits of the zero are zeroes,
and then take the square root.

\emph{The organization of the paper.} The paper is organized as follows. Three preliminary sections:
\begin{description}
\item[Section~\ref{sec:prelim}.] We define the proof systems and the measures of complexity we use.
\item[Section~\ref{sec:circuit}.] We define the circuit representation of polynomials appearing in an {\extpc} proof.
\item[Section~\ref{sec:bit}.] We define the bit arithmetic translation of circuits and prove useful facts about it.
\end{description}
Sections describing the three results:
\begin{description}
\item[Section~\ref{sec:ps}.] We prove the {\extpc+\ebvp} simulation of inequalities.
\item[Section~\ref{sec:sq}.] We prove that {\ebvp} allows {\extpc} to polynomially simulate the square root rule.
\item[Section~\ref{sec:cnf}.] We prove an exponential lower bound on the size of derivations of formulas in conjunctive normal form from {\ebvp}.
\end{description}
Eventually, in \textbf{Section~\ref{sec:oq}} we describe directions for further research.

\section{Preliminaries}\label{sec:prelim}
In this paper we work with polynomials over integers or rationals. 
We define \emph{the size of a polynomial} roughly as the total length of the bit representation of its coefficients.
Formally, let $f$ be an arbitrary integer or rational polynomial in variables $\{x_1, \ldots, x_n\}$. 
\begin{itemize}
    \item If $f \in \mathbb{Z}[x_1, \ldots, x_n]$ then $Size(f) = \sum (\lceil \log |a_i| \rceil + 1)$, where $a_i$ are the coefficients of $f$.
    \item If $f \in \mathbb{Q}[x_1, \ldots, x_n]$ then $Size(f) = \sum (\lceil \log |a_i| \rceil + \lceil \log|b_i| \rceil + 1)$, where $a_i \in \mathbb{Z}$, $b_i \in \mathbb{N}$ and $\frac{a_i}{b_i}$ are the coefficients of $f$.
\end{itemize}

We also use \emph{algebraic circuits}. 
Formally, an algebraic circuit is a dag whose vertices (gates) compute binary operations (addition and multiplication), thus gates have in-degree two; the inputs (or variables) and constants (nodes computing integers or rationals) 
are nodes of in-degree zero. Every gate of an algebraic circuit computes a polynomial in the input variables in a natural way; we sometimes identify a gate with the circuit consisting of all the nodes on which the gate depends (thus this gate is the output gate of such circuit). 

The size of the circuit is the number of its gates plus the sum of the bit sizes of all constants.
We will also be interested in the \emph{syntactic length} of an algebraic circuit, defined for circuits over $\mathbb{Z}$:
it is roughly a trivial upper bound on the number of bits of an integer computed by the circuit.
The definition essentially follows \cite{AGHT19_new}, augmenting it with the multiplication.

\begin{definition}[syntactic length of algebraic circuit]
Consider the gates of an algebraic circuit $G_1, \ldots, G_k$ in topological order. We define the syntactic length inductively:
\begin{itemize}
    \item If $G_i$ is an integer constant, then the syntactic length of $G_i$ is $\lceil\log(|G_i|)\rceil$.
    \item If $G_i = G_j + G_k$,  the syntactic length of $G_j$ is $t$, and the syntactic length of $G_k$ is $s$, then the syntactic length of $G_i$ equals $\max(s, t) + 1$.
    \item If $G_i = G_j \cdot G_k$,  the syntactic length of $G_j$ is $t$, and the syntactic length of $G_k$ is $s$, then the syntactic length of $G_i$ equals $s + t + 3$.
\end{itemize}
\begin{note}1. In the latter case the actual number of bits would be $s+t$; we state $s+t+3$ because this is how it is computed in our implementation of the integer multiplication in Sect.~\ref{sec:bit} --- however, it does not change much asymptotically, the resulting length changes at most polynomially.

2. Note that the circuit size cannot exceed its syntactic length.
\end{note}
\end{definition}
\subsection{Algebraic proof systems}
In what follows, $R$ denotes $\mathbb{Q}$ or $\mathbb{Z}$.

\begin{definition}[Polynomial Calculus, \cite{CEI96}]\label{def:PC}
Let $\Gamma = \{p_1, \ldots, p_m\} \subset R[x_1, \ldots, x_n]$ be a set of polynomials in variables $\{x_1, \ldots, x_n\}$ over $R$ such that the system of equations $p_1 = 0, \ldots, p_m = 0$ has no solution. A Polynomial Calculus ($\PC_R$) refutation of $\Gamma$ is a sequence of polynomials $r_1, \ldots, r_s$ where $r_s = const \neq 0$  and for every $l$ in $\{1, \ldots, s\}$, either $r_l \in \Gamma$ or $r_l$ is obtained through one of the following derivation rules for $j, k < l$:
\begin{itemize}
    \item $r_l = \alpha r_j + \beta r_k$, where $\alpha, \beta \in R$,
    \item $r_l = x_i r_k$.
\end{itemize}
The size of the refutation is $\sum_{l = 1}^s Size(r_l)$. The degree of the refutation is $\max_l deg(r_l)$.
\end{definition}

\begin{note}
1. In this paper we consider $\mathbb{Q}$ or $\mathbb{Z}$ as $R$ in $\PC_R$ above or ${\extpc}_R$ below. For both of these rings, we consider  the Boolean case, where axioms $x_i^2 - x_i = 0$ are present for every variable $x_i$, and for this case our proof systems are complete. 

2. Note that in the case $R=\mathbb{Q}$ one can assume $r_s=1$, while in the case $R=\mathbb{Z}$ an arbitrary nonzero constant is needed to maintain the completeness.
\end{note}

Tseitin's extension rule allows to introduce new variables for arbitrary formulas.
We use an algebraic version of this rule that allows to denote any polynomial by a new variable \cite{alekseev2020lower}.

\begin{definition}[Extended Polynomial Calculus, {\extpc}]\label{def:extpc}
Let $\Gamma = \{p_1, \ldots, p_m\} \subset R[x_1, \ldots, x_n]$ be a set of polynomials in variables $\{x_1, \ldots, x_n\}$ over $R$ such that the system of equations $p_1 = 0, \ldots, p_m = 0$ has no solution. An ${\extpc}_R$ refutation of $\Gamma$ is a Polynomial Calculus refutation of a set 
\begin{multline*}
\Gamma' = \{p_1, \ldots, p_m, y_1 - q_1(x_1, \ldots, x_n), y_2 - q_2(x_1, \ldots, x_n, y_1), \ldots, \\ y_m - q_m(x_1, \ldots, x_n, y_1, \ldots, y_{m - 1})\}
\end{multline*}
where $q_i \in R[\bar{x}, y_1, \ldots, y_{i - 1}]$ are arbitrary polynomials. 

We omit $R$ from the notation of $\PC_{R}$ or ${\extpc}_R$ when it is clear from the context.
The size of the {\extpc} refutation is equal to the size of the Polynomial Calculus refutation of $\Gamma'$.
\end{definition}

The square root rule \cite{GH03} allows to conclude that $f=0$ from $f^2=0$. We can consider it in the context of both {\PC} and {\extpc}.

\begin{definition}[{\pcsq}, {\extpcsq}]\label{def:PCS}
The proofs in {\pcsq}, {\extpcsq} follow
Definitions~\ref{def:PC},~\ref{def:extpc}
but allow one more derivation rule in terms of~Def.~\ref{def:PC}: 
\begin{itemize}
    \item derive $r_l$, if $r_l^2 = r_k$
\end{itemize}
(derive a polynomial if its square has been already derived).
\end{definition}

\begin{note}
If $R$ is a domain and $p^2 = 0$ for some $p \in R[\bar{x}]$, then $p = 0$.
\end{note}

The extended Binary Value Principle ({\ebvp}) says that that the (nonnegative) integer
value of a binary vector cannot be negative.
In order to use this fact in the proof, we need to specify
that such a polynomial can be replaced by $1$
(in particular, if {\ebvp} is present without a multiplier, it produces the contradiction $1=0$).

\begin{definition}[${\extpc}+{\ebvp}$]
${\extpc}+{\ebvp}$ operates exactly the same derivation rules as {\extpc} with one more rule:
\begin{itemize}
\item derive $r_l=g$ if for some polynomials $g, f_1,\ldots, f_t$ and integer constant $M>0$
we have derived the polynomial $r_k = g \cdot (M + f_1 + 2 f_2 + \ldots + 2^n f_t)$ along with polynomials
$
r_{k_1}=f_1^2 - f_1, \ldots, r_{k_t}=f_t^2 - f_t.
$
\end{itemize}
\end{definition}
\begin{note}
We can define ${\extpcsurd}+{\ebvp}$ the same way.
\end{note}

\subsection{A semialgebraic proof system}

We will consider the following proof system that can be viewed as a generalization of the {\LS} proof system \cite{Pud99} by the algebraic extension rule. Note that we could move the introduction of new variables to the beginning of the proof as we did in the definition of {\extpc}, however, it does not matter.

\begin{definition}[{\extps}]
Let $\Gamma = \{p_1, \ldots, p_m\} \subset R[x_1, \ldots, x_n]$ be a set of polynomials in variables $\{x_1, \ldots, x_n\}$ over $R$ such that the system of equations $p_1 \ge 0, \ldots, p_m \ge 0$ has no solution. An {\extps} refutation of $\Gamma$ is a sequence of polynomial inequalities $r_1 \ge 0, \ldots, r_m \ge 0$ where $r_m = -M$ ($M>0$ is an integer constant) and each inequality $r_l$ is obtained through one of the following inference rules:
\begin{itemize}
    \item $r_l = p_i$ for some $i$, or $r_l = x_i$, or $r_l = 1 - x_i$, or $r_l = x_i^2 - x_i$, or $r_l=x_i-x_i^2$, or $r_k = z^2$ for any variable $z$.  
    \item $r_l = r_i \cdot r_j$ or $r_l = r_i + r_j$ for $i, j < l$.
    (Note that we can infer $1$ as $x_i+(1-x_i)$, thus we can multiply by any positive constant.)

    \item If variable $y$ did not occur in polynomials $r_1, \ldots, r_{l - 1}$, then we can derive a pair of polynomials $r_l = y - f$, $r_{l + 1} = f - y$, where $f$ is one of the basic operations (addition, multiplication, copying) applied to variables not including $y$, and constants. 
\end{itemize}
Note that the newly introduced variables are not necessarily Boolean.
The size of the refutation is $\sum_{l = 1}^m Size(r_l)$. The degree of the refutation is $\max_l deg(r_l)$.

\end{definition}
\begin{note}
1. Once again, in the case $R=\mathbb{Q}$ we could assume $M=1$, while we need an arbitrary positive constant for $R=\mathbb{Z}$ in order to maintain completeness.

2. Note that while the definition of {\extps} is written in a slightly different manner compared to {\extpc}, it is not difficult to see that {\extps} polynomially simulates {\extpc} (in particular, conversion of equations to inequalities and of ideal inference to cone inference can be done similarly to \cite[Sect. 4.1.1 of the Technical Report version]{AGHT19_new}).
\end{note}


\section{Circuit and equational representations}\label{sec:circuit}
We will represent the polynomials of the {\extps} derivation as circuits in the input variables. In order to do this, we define circuit representations of axioms and extension variables.

\begin{definition}[Circuit representation: axioms]
For a polynomial $f\in\mathbb{Z}[\vx]$ appearing in the axiom $f\ge0$, we consider its arbitrary reasonable circuit representation
$$
Z_{f,1} = h_{f,1}(\vx), \ldots, Z_{f,s} = h_{f,s}(z_1, \ldots, z_{s-1}, \vx)
$$
where $h_{f,i}$ is one of the basic operations (addition, multiplication), a constant or one of the initial variables.
We denote the resulting circuit $Z_{f,s}$ by $Z_f$.
\end{definition}

We next define circuit representation for the extension variables. 

\begin{definition}[Circuit representation: extension variables]\label{def:circuitext}
Suppose we have a sequence of extension variables $y_1, \ldots, y_k$ introduced in some derivation by axioms $y_j = g_j(\vx, y_1, \ldots, y_{j - 1})$ (where $1 < j \le k$). We can define their values by algebraic circuits computed in a natural way (the axioms are substituted into each other): define the sequence of circuits $Y_1(\vx), \ldots, Y_k(\vx)$ by
\begin{itemize}
    \item $Y_1(\vx) = g_1(\vx)$,
    \item for each $1 < j \le k$, $Y_j(\vx) = g_j(\vx, Y_1(\vx), \ldots, Y_{j -1}(\vx))$.
\end{itemize}
We call $Y_i$ the circuit representation of the extension variable $y_i$.
\end{definition}

With the circuit representation of the extension variables and axioms, we can define the circuit representation of an {\extps} proof.

\begin{definition}[Circuit representation: {\extps} proof]
Given an {\extps} refutation $p_1 \ge 0, \ldots, -M = p_m \ge 0$ of a system in variables $x_i$, 
we construct the circuit representation $P_1,\ldots,P_m$
of its polynomials inductively:
\begin{itemize}
\item If $p_l$ is an axiom, $P_l$ is the circuit representation of this axiom.
\item If $p_l = x_i$, or $p_l = 1 - x_i$, then $P_l$ is the simple circuit computing $p_l$.
\item If $p_l = z^2$ for a variable $z$, then $P_l = Q\cdot Q$, where $Q$ is the circuit representation of $y$ (note that typically, $z$ is an extension variable).
\item If $p_l$ is obtained using a binary operation $\circ$ (addition or multiplication)
from $p_i$ and $p_j$, we put $P_l=P_i\circ P_j$.
\item If $p_l$ introduces a new variable, or it is the Boolean axiom $x_i^2-x_i$ (or $x_i-x_i^2$), we put $P_l=0$.
\end{itemize}
Note that the axioms and the extension variables appear in $P_i$'s
as subcircuits, and that the inputs of $P_i$'s correspond to the original variables of the system.
\end{definition}
\begin{definition}[equational representation]
Any algebraic circuit can be represented by equations (one equation per gate). More precisely, if we have gates $G_1, \ldots, G_m$ in  topological order, then we can consider variables $\gamma_1, \ldots, \gamma_m$ with the corresponding set of polynomial equations:
\begin{itemize}
    \item If $G_i = x_i$ or $1-x_i$ for some input variable, then corresponding polynomial equation for the $\gamma_i$ would be $\gamma_i = x_j$ or $\gamma_i = 1-x_i$.
    \item If $G_i = G_k \circ G_\ell$, then the corresponding polynomial equation for the $\gamma_i$ would be $\gamma_i = \gamma_k \circ \gamma_\ell$. 
\end{itemize}
We refer to this set of equations as the equational representation.
\end{definition}

The following lemma is used in the simulation of {\extps}.

\begin{lemma}\label{lem:circuit}
Consider the circuit and equational representations of an {\extps} proof $p_1\ge0,\ldots,p_t\ge0$.
Consider $P_i$ corresponding to the equational representation with the output variable $\pi_i$.
Then there is a polynomial-size (in the size of the original proof) {\extpc} derivation of $\pi_i=p_i$ using only the Boolean axioms
and the definitions of extension variables of the {\extps} proof.
The extension variables needed in the {\extpc} derivation
are those appearing in the equational representation.
\end{lemma}
\begin{proof}
First of all, note that if we consider any gate $Y_i$ from the circuit representation of the extension variables, then there is a polynomial-size {\extpc} proof of the equality
    $
    y_i = \upsilon_i,
    $
    where $\upsilon_i$ is the variable corresponding to the gate $Y_i$ in the equational representation. It follows by induction on the construction of the circuit representation (the sets of equations for the variables $\upsilon_i$ and $y_i$ are exactly the same).
    Similarly, for any variable $Z_{f}$ representing axiom $f$, there is a polynomial-size {\extpc} proof of the equality
    $
    f = \phi_f,
    $
    where variable $\phi_f$ corresponds to the gate $Z_{f}$. 
    
    Now we prove the statement of the lemma.
    We proceed by induction on the steps of the {\extps} proof:
    \begin{enumerate}
        \item If $p_l$ is an axiom, it follows from the discussion above.
        \item Recall that if $p_l$ introduces a new variable $y$, or it is the Boolean axiom $x_i^2-x_i$ (or $x_i-x_i^2$), we put $P_l=0$,
        that is, $\pi_l=0$ by definition. On the other hand, $p_l$ is an axiom for our {\extpc} proof, that is, $p_l=0$ is derived in a single step. Therefore $\pi_l = 0 = p_l$.
        \item If $p_l = x_i$, or $p_l = 1 - x_i$, then $P_l$ is the simple circuit computing $p_l$. Thus, it is also easy to prove that $\pi_l = p_l$.
        \item If $p_l = y^2$ for a variable $y$, then $P_l = Y\cdot Y$, where $Y$ is the circuit representation of the extension variable $y$ (if $y$ is the input variable, the situation is trivial). By the discussion above there is a polynomial-size derivation of 
        $
        \upsilon_i = y_i.
        $
        Then using the equation $\pi_l = \upsilon_i \cdot \upsilon_i$, we get that 
        $
        \pi_l = \upsilon_i \cdot \upsilon_i = y_i^2 = p_l.
        $
        \item If $p_l$ is obtained using a binary operation $\circ$ (addition or multiplication)
        from $p_i$ and $p_j$, we have $P_l=P_i\circ P_j$. Then the corresponding equation in the equational representation
        $
        \pi_l = \pi_i \circ \pi_j,
        $
        and we can use the induction assumption to  derive
        $
        \pi_l = \pi_i \circ \pi_j = p_i \circ p_j = p_l.
        $
    \end{enumerate}

\end{proof}

In order to simulate the square root derivation rule we need to consider a circuit representation of an arbitrary polynomial in extension variables, since a derivation in $\extpcsq$, unlike derivations in $\extps$, does not correspond to an algebraic circuit (algebraic circuits do not use square root gates).

\begin{definition}[Circuit representation: polynomials]\label{def:circuitpoly}
Consider a polynomial $g \in \mathbb{Z}[\vx, \vy]$, where $\vx$ are original variables and $\vy$ are variables introduced by the extension rule.  Def.~\ref{def:circuitext} defines the circuit representation $Y_1, \ldots, Y_m$ for the variables $y_1, \ldots, y_m$. Then we can consider any reasonable circuit  $G'_1, \ldots, G'_t$ computing the polynomial $g$ given variables $x_1, \ldots, x_n$, variables $y_1, \ldots, y_m$, and the constants. Substituting the subcircuits $Y_1, \ldots, Y_m$ in place of the inputs $y_1,\ldots,y_m$ of $G'_i$'s, we get the \emph{circuit representation} $G_1,\ldots,G_l$ of $g$.

The syntactic length of the polynomial $g$ is defined as the syntactic length of the circuit $G_1,\ldots,G_l$.
\end{definition}

The same proof works for a simplified version of Lemma~\ref{lem:circuit}:
\begin{lemma}\label{note:circuit}
 Consider any polynomial $g$ over the extension variables $y_1, \ldots, y_k$ and the original Boolean variables $x_1, \ldots, x_n$, and consider any reasonable circuit representation  $G'_1, \ldots, G'_t$ of $g$. Then we can substitute the subcircuits $Y_1, \ldots, Y_m$ in place of the inputs $y_1,\ldots,y_m$ of $G'_i$'s, and get the circuit representation $G_1,\ldots,G_l$ of $g$.

Then, if we consider an equational representation $\pi_1, \ldots, \pi_l$ of the circuit $G_1, \ldots, G_l$, then there is a polynomial-size (in the size of $g$) {\extpc} derivation of the equation
$$
g = \pi_l.
$$
\end{lemma}



\section{Explicit BIT definition and basic lemmas}\label{sec:bit}
In our {\extpc} simulations in Sections~\ref{sec:ps} and~\ref{sec:sq}, we argue about individual bits
of the values of the polynomials appearing in the {\extps} proof.
In this section we construct the circuits corresponding to these bits and prove auxilary statements about
our constructions.
We basically follow \cite{AGHT19_new} (Theorem 6.1 in the Technical Report version),
however, there are important differences:
\begin{enumerate}
    \item In the case of {\extpc} proofs, the circuits are used in the meta-language only.
    In the actual derivation,
    the bits are represented by extension variables
    defined through other extension variables, etc.
    (essentially computing the circuit value).
    \item Contrary to \cite{AGHT19_new}, we cannot magically switch
    between different representations of polynomials,
    every step of the derivation has to be done syntactically.
\end{enumerate}

The integers are represented in two's complement form (see the definition of $\val$ below). We use the following notation:
\begin{description}


\item[$\biti{i}(F)$:] if $F(\vx)$ is a circuit in the variables $\vx$, then $\biti{i}(F)$ is a new variable defined through other extension variables (and $\vx$) that computes the $i$-th bit of the integer computed by $F$ as a function of the input variables $\vx$, where the variables $\vx$ range over 0-1 values. The integer is represented in the two's complement form, that is, its highest bit is the sign bit.

\item[$\sgn(F)$] is used to denote this sign bit.


\item[$\bitv(F)$:] a collection of new variables that compute the bit vector of $F$. Note that $\bitv(F)$ also includes $\sgn(F)$.






\item[$\val(\vz)$:] the evaluation polynomial that converts bit encoding of an integer $\vz$ in two's complement representation to its integer value. Given $z_0,\ldots, z_{k-1}$,
$$
\val(\vz)=\sum_{i=0}^{k-2} 2^i\cdot z_i\,-2^{k-1}\cdot z_{k-1}.
$$

\end{description}

We construct the representation of  \biti{i}$(F)$ by induction on the size of $F$. 

\subsection{Proof strategy for the simulation}
Our plan for the simulation of {\extps} in Sect.~\ref{sec:ps} is as follows:
\begin{itemize}
    \item Suppose we have an {\extps} refutation $p_1(\vx, \vy) \ge 0, \ldots, p_m(\vx, \vy) \ge 0$, where $p_m = -M > 0$. We will consider the circuit representation $P_1, \ldots, P_m$ of polynomials $p_1, \ldots, p_m$ in order to speak about $\bitv(P_i)$,
    and will introduce more extension variables according to the corresponding equational representation of $P_i$'s.
    
    \item We will show by induction that we can derive the following statements in {\extpc}:
    \begin{enumerate}
        \item $\val(\bitv(P_i)) = p_i$.
        \item $\sgn(P_i) = 0$. 
    \end{enumerate}
    Then given the fact that $\val(\bitv(P_m)) = p_m = -M$, where $M \in \mathbb{N}$, and $\sgn(P_i) = 0$, we can apply $\ebvp$ to derive a contradiction in {\extpc}.
\end{itemize}
Before we accomplish this, we need to define BIT (using the definitions for basic arithmetic operation) and prove several useful lemmas about what can we derive in {\extpc} (basic facts about the values, the signs, etc). These will be also useful for the simulation of the square root rule in Sect.~\ref{sec:sq}.

\subsection{Basic arithmetic operations}
We now describe circuit constructions of the basic operations that we will need for the BIT definition.
A formal definition of those arithmetic operations essentially follows the scheme of \cite{AGHT19_new}. 
There is, however, one key difference: while in \cite{AGHT19_new} we defined the operations as circuits, in our context we define them as new variables alongside with their defining (sets of) equations. So all the capitalized notation above corresponds to new extension variables (sometimes with implicit introduction of auxilary extension variables) or vectors of new extension variables.

\begin{definition}[arithmetization operation $\arit{\cdot}$]\label{def:arithmetization}
 For a variable $x_i$, $\arit{x_i}:=x_i$. For the truth values false $\bot$ and true $\top$
we put $\arit{\bot}:=0$ and $\arit{\top}:=1$.
For logical connectives we define
$\arit{A\land B}:=\arit{A}\cdot \arit{B}$,
$\arit{A\lor B}:=1-(1-\arit{A})\cdot(1-\arit{B})$, and for the {\rm XOR} operation we define
$\arit{A\oplus B}:= \arit{A}+\arit{B}-2\cdot\arit{A}\cdot\arit{B}$.
\end{definition}

\begin{definition}[$\cari i$, $\addi i $, \addv]\label{def:carryadd}
When we use an adder for vectors of different size,
we pad the extra bits of the shorter one by its sign bit. Suppose that we have a pair of length-$(k + 1)$ vectors of variables $\vy=(y_0,\dots,y_{k}),\vz=(z_0,\dots,z_{k})$ of the same size. We first  pad the two vectors by a single additional bit $y_{k + 1}=y_{k}$ and $z_{k + 1}=z_{k}$, respectively (this is the  way to deal with  a possible overflow occurring  while adding the two vectors). Define
\begin{eqnarray*}
\cari{i}(\vy,\vz)&:=&
\begin{cases}
(y_{i-1}\land z_{i-1})\lor((y_{i-1}\lor z_{i-1})\land\cari{i-1}(\vy,\vz))
              ,               & i = 1,\ldots,k + 1;\\
              0\,, & i=0\,,
\end{cases}
\end{eqnarray*}
and
$$
\addi{i}(\vy,\vz):=
              y_i\oplus z_i \oplus \cari{i}(\vy,\vz)\,,                ~i = 0,\ldots,k.
$$
Finally, define
\[
\addv(\vy,\vz):=(\addi{t}(\vy,\vz),\cdots,\addi{0}(\vy,\vz))
\] (that is, $\addv$ is a multi-output circuit with $k+2$ output bits).
\end{definition}

\begin{definition}[absolute value operation \iabsv]\label{def:abs}
Let  $\vx$ be a $(k + 1)$-bit vector representing an integer in two's complement. Let  $s$ be  its sign bit, and let $\vm=\ve(s)$ be the $(k + 1)$-bit vector all of whose bits are $s$. Define $\iabsv(\vx)$ as the multi-output circuit that outputs $k+2$ bits as follows (where $\oplus$ here is bit-wise {\rm XOR}): $$\iabsv(\vx):=\addv(\vx,\vm)\oplus\vm.$$
\end{definition}

\begin{definition}[product of two nonnegative numbers in binary \prdvp]\label{def:prodp}
Let $\va$ be an $(r + 1)$-bit integer and $\vb$ be a $(k + 1)$-bit integer where the sign bit of both $\va,\vb$ is zero.
We define $k + 1$ iterations $i=0,\ldots,k$;
the result of the $i$-th iteration is defined as the $(r+i + 1)$-length vector $\overline s_i=s_{i,r+i}s_{i,r+i-1}\cdots s_{i,0}$, where
\begin{align*}
s_{ij}&:=a_{j-i}\land b_i, &\text{ for $i \le j\le r+i$},\\
s_{ij}&:=0 &\text{for $0\le j<i$.}
\end{align*}
(Note that we  use the sign bits $a_{k},b_{r}$ in this process although we assume it is zero; this is done in order to preserve uniformity with other parts of the construction.)
The  product of a $(k + 1)$-bit and an $(r + 1)$-bit integers is defined  as the sequential addition of all the results in all iterations:
\[
\prdvp(\va,\vb):=\addv\left(\overline s_{k},\addv\left(\overline s_{k-1},\ldots,\addv\left(\overline s_1,\overline s_0\right)\right)\ldots\right).
\]
The number of output bits of $\prdvp$ is formally $k + r + 2$ including the sign bit.
\end{definition}

\begin{definition}[product of two numbers in binary \prdv]\label{def:prod}
Let $\vy$ be an $(r + 1)$-bit integer and $\vz$ be a $(k + 1)$-bit integer in two's complement notation.
Define the product of $\vy$ and $\vz$ by first multiplying the absolute values of the two numbers
and then applying the corresponding sign bit:
\[
\prdv(\vy,\vz):= \addv\left(\prdvp\left(\iabsv(\vy),\iabsv(\vz)\right)\oplus\vm,s\right),
\]
where $s=y_{r}\oplus z_{k}$ and $\vm=\ve(s)$, with $y_{r},z_{k}$ the sign bits of $\vy,\vz$ as bit vectors in the two's complement notation, respectively.

Note that the number of bits that $\prdv$ outputs is  $k + r +5$:
given a $(k + 1)$-bit number, its $\iabs$ is of size $k+2$ (including the zero sign bit),
the nonnegative product $\prdvp$ of $\iabs(\vx)$ and $\iabs(\vy)$ has size $(k+2) + (r + 2)$,
bitwise XOR does not change the length, and adding $s$ augments the result by one more bit.
\end{definition}

\subsection{Definition of BIT}
Following \cite{AGHT19_new} we define the bit representation
of the values of polynomials computed by algebraic circuits.
In doing this, we construct another circuit. We identify its nodes
with new variables that will appear in our ${\extpc}+{\ebvp}$ proof,
and the defining equation for these variables are exactly the operations computed by the gates of the new circuit.
Note that the inputs of this circuit are the same as the inputs of the original circuit.

\begin{definition}[BIT]

Let $G_1 = f_1(\vx), G_2 = f_2(\vx, G_1), \ldots, G_m = f_m(\vx, G_1, \ldots, G_{m - 1})$ be a topological order of the gates of an algebraic circuit over variables $\vx$.

For each $G_r$ we define $\biti{i}(G_r)$ to be a \textbf{new extension variable} with the corresponding polynomial equation so that $\biti{i}(G_r)$ computes the $i$-th bit of $G_r$:

\noindent\textbf{Case 1:} $G_r = x_j$ for an input $x_j$. Then, $\biti{0}(G_r):=x_j$, $\biti{1}(G_r):=0$ (in this case there are just two bits).

\noindent\textbf{Case 2:} $G_r = \alpha$, for $\alpha\in \mathbb{Z}$. Then, $\biti{i}(G_r)$ is defined to be the $i$-th bit of $\alpha$ in two's complement notation.

\noindent\textbf{Case 3:} $G_r = G_k + G_l$. Then $\bitv(G_r) = \addv(\bitv(G_k),\bitv(G_l))$, and $\biti i(y_r)$ is defined to be the $i$-th bit of $\bitv(y_r)$.


\noindent\textbf{Case 4:} $G_r=G_k \cdot G_l$. Then $\bitv(G_r):=\prdv(\bitv (G_k),\bitv(G_l))$, and $\biti i(G_r)$ is defined to be the $i$-th bit of $\bitv(G_r)$.

Recall that in the latter two cases the shorter number is padded to match the length of the longer number by copying the sign bit before applying $\addv$ or $\prdv$.
\end{definition}

\subsection{The binary value lemma}
We now show a short proof of the fact that the BIT$(G)$ circuit that we constructed computes the same binary value
as the original circuit $G$. Moreover, it can be compactly proved in {\extpc} for the equational representation of BIT$(G)$.


\begin{lemma}[binary value lemma]\label{lem:binary-value-lemma}
Let $y_1 = f_1(\vx), y_2 = f_2(\vx, y_1), \ldots, y_m = f_m(\vx, y_1, \ldots, y_{m - 1})$ be the equational representation of the algebraic circuit \[G_1(\vx) = f_1(\vx), \ldots, G_m = f_m(\vx, G_1(\vx), \ldots, G_{m - 1}(\vx))\] over the variables $\vx=\{x_1,\dots,x_n\}$, and let t be the syntactic length of $G_1, \ldots, G_m$.

Then, there is an {\extpc} proof (using only the Boolean axioms and the equations of the BIT encoding) of $$y_i=\val(\bitv(G_i))$$ of size $\poly(t)$ for each $1 \le i \le m$.
\end{lemma}

\begin{proof}
For the proof we refer to the similar lemma from \cite{AGHT19_new}. That paper talks about another system, {\IPS},
which incorporates polynomial identity testing for free. However, the proof of this lemma is syntactic and does not use
polynomial identity testing.
We will briefly describe the structure of the proof. 

The proof proceed by induction. On each induction step we assume that we have already constructed {\extpc} proofs for the equations 
$$
y_1 = \val(\bitv(G_1)), \ldots, y_r = \val(\bitv(G_r))
$$
and construct the proof of the equation $y_{r + 1} = \val(\bitv(G_{r + 1}))$. The construction of the {\extpc} proof depends on the way in which the variable $y_{r + 1}$ was introduced. For example, if $y_{r + 1} = y_k \cdot y_l$, then $G_{r + 1}$ is a product gate and $G_{r + 1} = G_k \cdot G_l$. We need to show that 
$$
\val(\prdv(\bitv(G_k), \bitv(G_l))) = \val(\bitv(G_k)) \cdot \val(\bitv(G_l)),
$$ 
which can be done exactly in the same way as in  \cite{AGHT19_new}. 
\end{proof}

\subsection{Useful lemmas about the BIT value}
In this section we describe technical lemmas about individual bits in the bit representation that will be used later in the proof of our simulation.


\begin{lemma}\label{lem:bits-of-zero}
For any vector of variables $r_0, \ldots, r_{k - 1}, r_{k}$, there is a
$\poly(k)$-size ${\extpc}+{\ebvp}$ derivation of 
$$
r_0 = \ldots = r_{k}= 0
$$ from
$$
r_0^2 - r_0 = 0, \ldots, r_{k}^2 - r_{k} = 0 \text{ and } r_0 + 2 r_1 + \ldots + 2^{k - 1} r_{k - 1} - 2^{k} r_{k} = 0.$$
\end{lemma}
\begin{proof}
Multiply the last equation by $r_{k}$ and replace $r_{k}^2$ by $r_{k}$.
We get $(r_0 + 2 r_1 + \ldots + 2^{k - 1} r_{k - 1} - 2^{k} ) r_{k} = 0,$
which has (the negation of) an instance of {\ebvp} in the parentheses (for $r'_i = 1-r_i$).
It remains to apply the {\ebvp} rule to prove that $r_{k} = 0$. After that we get
$$
r_0 + 2 r_1 + \ldots + 2^{k - 1} r_{k - 1} = 0.
$$
Again, multiply this by $r_{k - 1}$ and replace $r_{k - 1}^2$ by $r_{k - 1}$.
We get $(r_0 + 2 r_1 + \ldots + 2^{k - 2} r_{k - 2} + 2^{k - 1}) r_{k - 1} = 0$
with an instance of {\ebvp} inside. After applying the {\ebvp} rule we get that $r_{k - 1} = 0$. We can continue in the same way for $r_{k - 2}, \ldots, r_0$ getting
$$
r_0 = \ldots = r_{k}= 0.
$$
\end{proof}

\begin{lemma}[monotonicity of addition and multiplication]\label{sign of prod and sum}
For any two bit vectors $r_0, \ldots, r_{k - 1}, r_{k}$ and  $r_0', \ldots, r_{k - 1}', r_{k}'$, 
there is a $\poly(k)$-size {\extpc} derivation of
$$
\sgn(\prdv(\overline{r}, \overline{r}')) = 0 \text{ and } \sgn(\addv(\overline{r}, \overline{r}')) = 0,
$$
from
\begin{align*}
r_0^2 - r_0 = 0, \ldots, r_{k - 1}^2 - r_{k - 1} = 0, r_{k}^2 - r_{k} = 0, \\ {r_0'}^2 - r_0' = 0, \ldots, {r_{k - 1}'}^2 - r_{k - 1}' = 0, {r_{k}'}^2 - r_{k}' = 0,  \\   r_{k} = 0, \\    r_{k}' = 0.
\end{align*}
\end{lemma}
\begin{proof}
See \cite{AGHT19_new} (Lemma 6.7 in the Technical Report version), as the derivation presented in that paper is literally in {\extpc}. 
\end{proof}


\begin{lemma}\label{lem:square}
1. For any vector of variables $r_0, \ldots, r_{k - 1}, r_k$,
there is a $\poly(k)$-size {\extpc} derivation of
$$
\sgn(\prdv(\overline{r}, \overline{r})) = 0
$$
from
$$
r_0^2 - r_0 = 0, \ldots, r_{k - 1}^2 - r_{k - 1} = 0, r_{k}^2 - r_{k} = 0.
$$
2. If additionally $\prdv(\overline{r}, \overline{r}) = \overline{0}$ is given,
there is a $\poly(k)$-size {\extpc} derivation of
$$
r_0 = 0, \ldots, r_k = 0. 
$$
\end{lemma}
\begin{proof}
By the definition of $\prdv$, \[
\prdv(\vr,\vr) = \addv\left(\prdvp\left(\iabsv(\vr),\iabsv(\vr)\right)\oplus\vm,s\right),
\]
where $s=r_{k}\oplus r_{k}$ and $\vm=\ve(s)$. Thus we instantly derive that $s = 0$ and $\vm = \overline{0}$ and obtain 
$$
\prdv(\vr,\vr) = \prdvp\left(\iabsv(\vr),\iabsv(\vr)\right),
$$
which completes the proof of the first statement (by the definition of $\prdvp$).

Now we denote $\overline{r}' := \iabsv(\vr)$. We already know from the definition of {\iabsv} that the sign bit of $\vr'$ is equal to 0. Now we will derive that each bit $r_i$ is equal to zero by induction, starting from $r_0$. 

{\textbf{Base case:}} We have the equation $\prdvp(\vr', \vr') = \overline{0}$. Let us denote the vector $\prdvp(\vr', \vr')$ as $\overline{t}$.

Now recall the definition of $\prdvp$: we have $k + 1$ iterations $i = 0, \ldots,k$; the result of the $i$th iteration is defined as the $(k+ i + 1)$-length vector $\overline s_i=s_{i,k+i}s_{i,k+i-1}\cdots s_{i,0}$ where
\begin{align*}
s_{ij}&:=r_{j-i}'\land r_i', &\text{ for $i \le j\le k + i$},\\
s_{ij}&:=0 &\text{for $0\le j<i$.}
\end{align*}
Eventually, {\prdvp} is defined as
\[
\overline{t} := \addv\left(\overline s_{k},\addv\left(\overline s_{k - 1},\ldots,\addv\left(\overline s_1,\overline s_0\right)\right)\ldots\right).
\]
From this definition, it is immediate that $t_0 = r_0'$ since $s_{i, 0} = 0$ for $i > 0$ (which matches the intuition of the ``school'' multiplication procedure). So, we can easily derive that $r_0' = 0$. 

{\textbf{Induction step:}} Assume we already derived that $r_l' = 0, \ldots, r_0' = 0$. After substituting these values, the definition of $\bar s_i$ gives us immediately 
\begin{itemize}
    \item 
$\bar{s}_l = \bar{s}_{l - 1} = \ldots = \bar{s}_0 = 0$,
\item
$s_{i j} = 0$ for $0 \le j \le l$, $i > l$,
\item
$s_{i j} = 0$ for $l+ 1 \le i$ and $l + 1 \le j < 2l + 2$,
\item
thus we can conclude that $s_{i j} = 0$ for \emph{any} $i$ and $0 \le j < 2l + 2$.
\end{itemize}


Finally, for $j = 2l + 2$ we can derive that $s_{ij} = 0$ for all $i > l + 1$ because it is either defined to be 0 or $s_{i, j} = r_{j-i}'\land r_i'$ and $r_{j - i}' = 0$ was derived already (since $j - i \le l$). Also $\bar{s}_i=0$ for $i\le l$, so $s_{i, 2 l + 2} = 0$ for $i \neq l + 1$. Together with the fact that $s_{i j} = 0$ for any $j < 2l + 2$ we can derive that 
$$
t_{2l + 2} = s_{l + 1, 2l + 2}
$$
(we use here the definition of ADD, which is ``school'' addition, and we have just obtained that not only all the bits in the column $2l+2$ are zeroes, but also every bit in less significant columns is zero).

On the other hand, by definition
$
s_{l + 1, 2l + 2} = (r_{l + 1}' \land r_{l + 1}') = r_{l + 1}',
$
so we conclude that
$
t_{2l + 2} =  r_{l + 1}',
$
which gives us $r_{l + 1}' = 0$.\

Thus we have shown that $\iabsv(\vr) = \overline{0}$. Now using a simple induction argument again we can show that $\vr = \overline{0}$.
\end{proof}

\section{${\extpcZ} + {\ebvp}$ polynomially simulates ${\extpsZ}$}\label{sec:ps}
In this section we will show that ${\extpcZ} + {\ebvp}$ polynomially simulates ${\extpsZ}$. This will be done by gradually applying Lemma~\ref{lem:binary-value-lemma} to the circuit representation of the ${\extpsZ}$ derivation.

\begin{theorem}[the derivation theorem]\label{Main_BV_lemma}
Suppose we have a system of polynomial equations  $f_1 = 0, \ldots, f_k = 0$,
and that there is an ${\extpsZ}$ refutation $p_1\ge0,\ldots,p_m\ge0$ of
the corresponding system  $f_1 \ge 0, f_1 \le 0, \ldots, f_k \ge 0, f_k \le 0$. 

Consider its circuit representation according to Sect.~\ref{sec:circuit}. Denote the syntactic length of the circuit $P_1, \ldots, P_m$ as $t$. Then, in terms of Sect.~\ref{sec:circuit} there are $\poly(t)$-size $\extpc_\mathbb{Z}$ + {\ebvp} derivations of the facts
\begin{enumerate}
    \item $p_1 = \val(\bitv(P_1)), \ldots, p_m = \val(\bitv(P_m)).$
    \item Each \textbf{sign bit} in $\bitv(P_i)$ is equal to 0 (written in the form of polynomial equation $s_i = 0$ where $s_i$ is a variable, corresponding to the sign bit of $\bitv(P_i)$). 
\end{enumerate}
The axioms used in these derivations are the boolean axioms, the axioms defining extension variables, and (for the second statement) the input axioms.
\end{theorem}
\begin{proof}
\begin{enumerate}
    \item  From Lemma~\ref{lem:binary-value-lemma} we know a short proof that the binary value of the BIT circuit BIT$(G)$ equals the variable corresponding to the output of the original circuit $G$ in the equational representation of $G$. By applying this lemma to all circuits appearing in the proof we get $\pi_i = \val(\bitv(P_i))$, where the variable $\pi_i$ corresponds to the output of $P_i$. It remains to prove the equation $\pi_i=p_i$, which is done by Lemma~\ref{lem:circuit}.




\item In order to prove that there are polynomial-size derivations of the facts that each sign bit in $\bitv(P_i)$ is equal to 0, we recall that previously proven lemmas give us three statements:
\begin{enumerate}
    \item\label{sign from val statement} If we have the equation $\val(\bitv(P_j)) = 0$, then Lemma~\ref{lem:bits-of-zero} provides a polynomial-size derivation of $\sgn(\bitv(P_j)) = 0$.
    \item\label{sum and prod statement} If we have equations $\sgn(\bitv(P_j)) = 0$ and $\sgn(\bitv(P_k)) = 0$, then Lemma~\ref{sign of prod and sum} provides a polynomial-size derivation of 
    $$
    \sgn(\prdv(\bitv(P_j), \bitv(P_k))) = 0 \text{ and } \sgn(\addv(\bitv(P_j),\bitv( P_k))) = 0.
    $$
    \item\label{square statement} For any variable $y_i$, Lemma~\ref{lem:square} provides a polynomial-size derivation of
    $$
    \sgn(\prdv(\bitv(Y_i), \bitv(Y_i))) = 0.
    $$
\end{enumerate}
We now proceed to proving the statement 2 by induction.

\textbf{Base case:} the base is one of the following cases:
\begin{itemize}
    \item $P_i$ is a definition of an ${\extpsZ}$ proof extension variable or a Boolean axiom. Then $P_i = 0$ (that is, it is a trivial circuit) by the construction of the circuit representation (cf. Lemma~\ref{lem:circuit}, second item in the proof). 
    \item $P_i$ is an input axiom. By the first statement we derive $\val(\bitv(P_i)) = 0$ and using \hyperref[sign from val statement]{statement (a)}, we can derive that  $\sgn(\bitv(P_i)) = 0$.
    \item $P_i$ is an input variable or its negation; then $\sgn(\bitv(P_i)) = 0$ is easily seen from the construction of $\bitv$.
    \item $P_i$ is a square (of a variable). Then \hyperref[square statement]{statement (c)} provides a polynomial-size derivation of $\sgn(\bitv(P_i)) = 0$.
\end{itemize}

\textbf{Induction step:} Suppose we have already proved that $\sgn(\bitv(P_j)) = 0$ for $j < k$, and $P_k$ is constructed using an operation $P_{k} = P_j \cdot P_l$ or $P_{k} = P_j + P_l$. Then
we can apply \hyperref[sum and prod statement]{statement (b)} and show that $\sgn(\bitv(y_{k + 1}'')) = 0$ with polynomial-size $\extpc$ derivation.\end{enumerate}
\end{proof}

\subsection{The simulation theorem}
\begin{definition}[Syntactic size of a refutation]
The syntactic size of an ${\extpsZ}$ refutation is the syntactic size of a corresponding circuit representation from Sect.~\ref{sec:circuit}.
\end{definition}

\begin{theorem}\label{th:ext-pc-sim}
Consider arbitrary system of polynomial equations $f_1 = 0, \ldots, f_k = 0$. Suppose there is an ${\extpsZ}$ refutation for the system $f_1 \ge 0, f_1 \le 0, \ldots, f_k \ge 0, f_k \le 0$ of syntactic size $S$. Then there is an $\extpc_{\mathbb{Z}}+\ebvp$ refutation for the system $f_1 = 0, \ldots, f_k = 0$ of size at most $\poly(S)$.
\end{theorem}
\begin{proof}
We use the notation from the previous section.

Consider a size $S$ {\extps}-refutation $p_1 \ge 0, \ldots, -M = p_k \ge 0$ of the system $f_1 \ge 0, f_1 \le 0, \ldots, f_k \ge 0, f_k \le 0$. By \autoref{Main_BV_lemma}(1) there is a $\poly(S)$ derivation of the fact
that the value of the polynomial computed in the last line ($p_k\ge0$, which is $-M\ge0$) of the original semialgebraic proof is a negative integer 
$$
-M = p_k = \val(\bitv(P_k)). 
$$
On the other hand, by \autoref{Main_BV_lemma}(2) there is a $\poly(S)$ derivation of the fact that
$$
s = 0,
$$
where $s$ is a variable corresponding to the sign bit of $\bitv(P_k)$.
This means that we have an equation of the form 
$$
-M = b_0 + 2 b_1 + 4 b_2 + \ldots + 2^r b_r - 2^{r + 1} s
$$
where $b_0, \ldots, b_r, s$ are the variables corresponding to the bit representation of $\bitv(P_k)$.
From this we derive that 
$$
b_0 + 2 b_1 + 4 b_2 + \ldots + 2^r b_r + M = 0,
$$
which is exactly the case of {\ebvp}, so the contradiction follows in a single step.
(Note that another application of {\ebvp} is in Lemma~\ref{lem:bits-of-zero}.)
\end{proof}


\section{${\extpcZ} + {\ebvp}$ polynomially simulates  ${\extpcsurdz} + {\ebvp}$}\label{sec:sq}
In this section we show that {\ebvp} simulates the square root rule.

We will be using the following strategy for the simulation:
\begin{itemize}
    \item Suppose we want to derive $g = 0$ from $g^2 = 0$, for some polynomial $g$.
    \item We consider the bit representation $\bitv(G^2)$ of $g^2$.
    \item Lemma~\ref{lem:binary-value-lemma} provides a polynomial-size proof of
    $
    \val(\bitv(G^2)) = g^2,
    $
    thus we have 
    $
    \val(\bitv(G^2)) = 0.
    $
    \item From this, Lemma~\ref{lem:bits-of-zero} provides a polynomial-size proof of $\bitv(G^2) = \overline{0}$. Here we make use of {\ebvp}.
    \item Now Lemma~\ref{lem:square} provides a polynomial-size proof of $\bitv(G) = \overline{0}$.
    \item From this we can derive that $g  = \val(\bitv(G)) = 0$.
\end{itemize}

The formal application of this strategy is given by the following lemma.

\begin{lemma}\label{squre root simutaion lemma}
Assume that we have a polynomial $g \in \mathbb{Z}[x_1, \ldots, x_n, y_1, \ldots, y_m]$ where $x_1, \ldots, x_n$ are Boolean variables (that is, we have the equations $x_i^2 - x_i = 0$), and variables $y_1, \ldots, y_m$ are other variables introduced via the extension rule (which means that each $y_j = h_j(\vx, y_1, \ldots, y_{j - 1})$, where $h_j$ is a basic arithmetic operation or a constant). Suppose the syntactic length (cf Def.~\ref{def:circuitpoly} of the polynomial $g$ is $t$. Then there is a $\poly(t)$-size ${\extpcZ}+{\ebvp}$ derivation of the equation $g = 0$ from the equation $g^2 = 0$ (using the equations $x_i^2 - x_i = 0$ and $y_j - h_j(\vx, \vy) = 0$).
\end{lemma}
\begin{proof}
Consider the circuit representation of the polynomial $g$. We can now consider the BIT representation of this circuit,
and get (by Lemma~\ref{lem:binary-value-lemma}) a polynomial-size derivation of
$$
\pi_1 = \val(\bitv(G_1)), \ldots, \pi_l = \val(\bitv(G_l)).
$$
On the other hand, we can apply Lemma~\ref{note:circuit} to prove that
$$
g = \pi_l.
$$
Let us add one more gate $G_{l + 1}$ to the circuit:   $G_{l + 1} = G_l \cdot G_l$. The corresponding variable in the equational representation would be
$
\pi_{l + 1} = \pi_l \cdot \pi_l.
$
Then we can instantly derive from $g^2 = 0$ that 
$$
\pi_{l + 1} = g^2 = 0. 
$$
Thus, using the equation $\pi_{l + 1} = \val(\bitv(G_{l + 1}))$ we can derive that
$$
\val(\bitv(G_{l + 1})) = 0.
$$
Lemma~\ref{lem:bits-of-zero} (that uses {\ebvp}) allows us to derive
$$
\bitv(G_{l + 1}) = \overline{0}.
$$
Now using the fact that $\bitv(G_{l + 1}) = \prdv(\bitv(G_l), \bitv(G_l))$ and Lemma~\ref{lem:square}, we can derive that
$$
\bitv(G_{l}) = \overline{0}.
$$
Now, using the equation $g = \val(\bitv(G_{l}))$ we instantly get that $g = 0$.
\end{proof}

We can now state the simulation result.

\begin{theorem}\label{th:sq-root}
Consider arbitrary system of polynomial equations $f_1 = 0, \ldots, f_k = 0$. Suppose there is an $\extpcsurdz+\ebvp$ refutation for this system where the sum of the syntactic sizes of all polynomials in derivation is equal to $S$. Then there is an $\extpc_{\mathbb{Z}}+\ebvp$ refutation for the system $f_1 = 0, \ldots, f_k = 0$ of size at most $\poly(S)$.
\end{theorem}
\begin{proof}
We proceed by induction. Assume that we constructed an $\extpc+\ebvp$ derivation of polynomials $p_1, \ldots, p_l$ that appeared in the original $\extpcsurd+\ebvp$ derivation. We now show how to derive the polynomial $p_{l + 1}$. If this polynomial is an axiom, or is derived by the extension rule, or is derived by the {\ebvp} rule, or is derived by addition or multiplication from previous polynomials, then we can derive it in {\extpc} using the same rule (note that the size of the derivation is always at most the syntactic length). If the polynomial was derived by the square root rule, then we can use Lemma~\ref{squre root simutaion lemma} to simulate this derivation.
\end{proof}


\section{{\ebvp} cannot be used to prove CNF lower bounds}\label{sec:cnf}
Exponential lower bounds on the size of proofs of {\ebvp} have been demonstrated for several proof systems including {\extpcsq} \cite{alekseev2020lower}. However, they have a caveat: {\ebvp} is not a translation of a Boolean formula in CNF. Is it still possible to use these bounds to prove an exponential lower bound for a formula in CNF? For example, one could provide a polynomial-size {\extpcsq} derivation of a translation of an unsatisfiable Boolean formula in CNF from {\ebvp}: together with the lower bound for {\ebvp}, this would prove a bound for a formula in CNF. One could even introduce extension variables in order to describe such a formula.

In this section we show that this is not possible: any {\extpcsq} derivation of an unsatisfiable CNF from ${\ebvp}_n$ (that is, from $\sum_{i=1}^{n} x_i2^{i-1} + M = 0$) should have exponential size in $n$. We start with proving a lower bound over the integers. Then we use this result to extend it to the rationals. The proof can be viewed as a generalization of the lower bound in \cite{alekseev2020lower}; however, the lower bound is proved not for the derivation of $M=0$, but for the derivation of an arbitrary unsatisfiable CNF, possibly in the extension variables.

\subsection{Lower bound over the integers}
Suppose we have derived some unsatisfiable formula in CNF from $\ebvp_n$ in {\extpcsqZ}. This means that we have derived polynomial equations of the form $C_1 \cdot p_1 = 0, \ldots, C_m \cdot p_m = 0$, where each $C_i$ is a nonzero integer constant and each $p_i$ is the translation of a Boolean clause. The translation has the following form:
$$
p_i = y_{j_1} \cdots y_{j_k} \cdot \neg{y}_{\ell_1} \cdots \neg{y}_{\ell_r},
$$
where each $y_j$ is a Boolean variable and $\neg{y}_{\ell}$ is a variable introduced via the extension rule $\neg{y}_{\ell} = 1 - y_{\ell}$. Note that each variable $y_j$ can be an extension variable, however, it is necessary that we should derive that $C_j' \cdot (y_j^2 - y_j) = 0$ for each $y_j$, where $C_j' \in \mathbb{Z} \backslash \{0\}$. We will fix those equations $C_j' \cdot (y_j^2 - y_j) = 0$ for later. 
Note that since we work over the integers, we cannot assume that all $C_i$'s and $C'_j$'s equal 1 (we cannot divide), though if we derive polynomials multiplied by nonzero constants, it may still help in proving a lower bound for a CNF.

We start with formally defining how a substitution into the input variables changes polynomials that use extension variables:
\begin{definition}
Suppose we have introduced variables $y_1, \ldots, y_m$ in an {\extpcsqZ} derivation as
$$
 y_1 = q_1(x_1, \ldots, x_n), y_2 = q_2(x_1, \ldots, x_n, y_1), \ldots, \\ y_m = q_m(x_1, \ldots, x_n, y_1, \ldots, y_{m - 1}).
$$
Then, for any variable $y_i$ and any vector of bit values $\{b_1, \ldots, b_n\} \in \{0, 1\}^n$ we can define substitution $y_i|_{x_1 = b_1, \ldots, x_n = b_n}$ in the following way:
\begin{itemize}
    \item $y_1|_{x_1 = b_1, \ldots, x_n = b_n} := q_1(b_1, \ldots, b_n)$.
    \item For $i > 1$ we define 
    $$
    y_i|_{x_1 = b_1, \ldots, x_n = b_n} := q_i(b_1, \ldots, b_n, y_1|_{x_1 = b_1, \ldots, x_n = b_n}, \ldots, y_{i - 1}|_{x_1 = b_1, \ldots, x_n = b_n}).
    $$
\end{itemize}
For any polynomial $f(x_1, \ldots, x_n, y_1, \ldots, y_m) \in \mathbb{Z}[\vx, \vy]$ we define $f|_{x_1 = b_1, \ldots, x_n = b_n}$ in the following way:
$$
f|_{x_1 = b_1, \ldots, x_n = b_n} = f(b_1, \ldots, b_n, y_1|_{x_1 = b_1, \ldots, x_n = b_n}, \ldots, y_m|_{x_1 = b_1, \ldots, x_n = b_n})
$$
\end{definition}

Before proving our lower bound, we observe a property of Boolean substitutions:
\begin{lemma}\label{lem:primediv}
Suppose we have an instance of {\ebvp} of the form $M + x_1 + 2 x_2 + \ldots + 2^{n - 1} x_n$. Consider any prime number $p < 2^n$ and the binary representation $b_1, \ldots, b_k$ of any number $0 \le t < 2^n$ such that  $t \equiv -M\ (\text{mod } p)$. Suppose we have an {\extpcsqZ} derivation of the polynomial equation $f = 0$ from $M + x_1 + 2 x_2 + \ldots + 2^{n - 1} x_n$ and the Boolean axioms $x^2_i-x_i=0$. Then the number 
$f|_{x_1 = b_1, \ldots, x_n = b_n}$ is divisible by $p$.
\end{lemma}
\begin{proof}
The proof of this statement is a straightforward induction. It is obvious that the integers $$
M + b_1 + 2 b_2 + \ldots + 2^{n - 1} b_n ,
\quad b_i^2 - b_i = 0 \text{ and }(y_i - q_i)|_{x_1 = b_1, \ldots, x_n = b_n} = 0
$$
are divisible by $p$. Now we will prove the induction step:
\begin{itemize}
    \item If we have any derivation of the form $f_l = \alpha f_j + \beta f_k$, where $\alpha, \beta \in \mathbb{Z}$, then $f_k|_{x_1 = b_1, \ldots, x_n = b_n}$ and $f_j|_{x_1 = b_1, \ldots, x_n = b_n}$ are divisible by $p$, so $f_l|_{x_1 = b_1, \ldots, x_n = b_n}$ is divisible by $p$.
    \item If $f_l = x_j f_k$ or $f_l = y_j f_k$, then \begin{align*}
        f_l|_{x_1 = b_1, \ldots, x_n = b_n} = b_j f_k|_{x_1 = b_1, \ldots, x_n = b_n} \text{ or } \\ f_l|_{x_1 = b_1, \ldots, x_n = b_n} = y_j|_{x_1 = b_1, \ldots, x_n = b_n} \cdot f_k|_{x_1 = b_1, \ldots, x_n = b_n},
    \end{align*}
    so $f_l|_{x_1 = b_1, \ldots, x_n = b_n}$ is divisible by $p$.
    \item $f_l^2 = f_k$, then since $p$ is prime and $f_k|_{x_1 = b_1, \ldots, x_n = b_n}$ is divisible by $p$, $f_l$ also should be divisible by $p$.
\end{itemize}
\end{proof}
Immediately we get the following corollary:
\begin{corollary}\label{cor:primediv}
Suppose we have an instance of {\ebvp} of the form $M + x_1 + 2 x_2 + \ldots + 2^{n - 1} x_n$. Consider any prime number $p < 2^n$ and the binary representation $b_1, \ldots, b_k$ of any number $0 \le t < 2^n$ such that  $t \equiv -M\ (\text{mod } p)$. Suppose we introduced extension variable $y_i$ for which we have an {\extpcsqZ} derivation of the polynomial equation $C' \cdot (y_i^2 - y_i) = 0$ from $M + x_1 + 2 x_2 + \ldots + 2^{n - 1} x_n$. Then, either the number $C'$ is divisible by $p$, or $y_i|_{x_1 = b_1, \ldots, x_n = b_n} \equiv 1\ (\text{mod } p)$, or $y_i|_{x_1 = b_1, \ldots, x_n = b_n} \equiv 0\ (\text{mod } p)$.
\end{corollary}
\begin{proof}
Straightforward from Lemma~\ref{lem:primediv}. We know that $C' \cdot (y_i|_{x_1 = b_1, \ldots, x_n = b_n}^2 - y_i|_{x_1 = b_1, \ldots, x_n = b_n})$ is divisible by $p$. Then, either $C'$ or $(y_i|_{x_1 = b_1, \ldots, x_n = b_n}^2 - y_i|_{x_1 = b_1, \ldots, x_n = b_n})$ is divisible by $p$. If $(y_i|_{x_1 = b_1, \ldots, x_n = b_n}^2 - y_i|_{x_1 = b_1, \ldots, x_n = b_n})$ is divisible by $p$, then either $y_i|_{x_1 = b_1, \ldots, x_n = b_n} \equiv 1\ (\text{mod } p)$, or $y_i|_{x_1 = b_1, \ldots, x_n = b_n} \equiv 0\ (\text{mod } p)$.
\end{proof}
Now we are ready to prove an exponential lower bound over the integers:
\begin{theorem}\label{thm:lower_bound_integers}
Suppose we have an {\extpcsqZ} derivation of an unsatisfiable CNF from $M + x_1 + \ldots + 2^{n - 1}x_n=0$ and the Boolean axioms. Then at least one of the following three conditions holds:
\begin{itemize}
    \item The number of clauses in this CNF is at least $2^{n/3}$. 
    \item We have derived a polynomial equation $C' \cdot (y_j^2 - y_j) = 0$ and the constant $C'$ is divisible by at least $\Omega(2^{n/3})$ different prime numbers.
    \item There is a clause $C \cdot y_{j_1} \cdots y_{j_k} \cdot \neg{y}_{\ell_1} \cdots \neg{y}_{\ell_r}$ such that the constant $C$ is divisible by at least $\Omega(2^{n/3})$ different prime numbers.
\end{itemize}
\end{theorem}
\begin{proof}
Let $\mathcal{Y}$ be the set of variables occurring in our CNF. 

Consider the set $\mathcal{P}$ of all prime numbers from $\{1, 2, \ldots, 2^n - 1\}$. Now consider any prime number $p \in \mathcal{P}$. As in Lemma~\ref{lem:primediv}, we can take an arbitrary  $t \in \mathbb{Z}$, $0 \le t < 2^n$, such that $t \equiv -M \ (\text{mod } p)$. Consider the binary representation $b_1, \ldots, b_n$ of this integer $t$. Corollary~\ref{cor:primediv} says that for every $y_i \in \mathcal{Y}$ we have derived that $C'_i \cdot (y_i^2 - y_i) = 0$ and either $C'_i$ is divisible by $p$, or $y_i|_{x_1 = b_1, \ldots, x_n = b_n} \equiv 1\ (\text{mod } p)$, or $y_i|_{x_1 = b_1, \ldots, x_n = b_n} \equiv 0\ (\text{mod } p)$. We fix now this particular equation for $y_i$ in what follows.

Now suppose that for every $y_i \in \mathcal{Y}$, the constant $C_i'$ from equation  $C_i' \cdot (y_i^2 - y_i) = 0$ is not divisible by $p$. Then we know that every number $y_i|_{x_1 = b_1, \ldots, x_n = b_n}$ is Boolean modulo $p$. Thus every number $\neg{y_i}|_{x_1 = b_1, \ldots, x_n = b_n}$ is also Boolean modulo $p$ and 
$$
y_i|_{x_1 = b_1, \ldots, x_n = b_n} \equiv 1 - \neg{y_i}|_{x_1 = b_1, \ldots, x_n = b_n}\ (\text{mod }p).
$$
Then, since our CNF is \emph{unsatisfiable}, we know that there is a clause  $C \cdot y_{j_1} \cdots y_{j_k} \cdot \neg{y}_{\ell_1} \cdots \neg{y}_{\ell_r}$, such that 
$$
(y_{j_1} \cdots y_{j_k} \cdot \neg{y}_{\ell_1} \cdots \neg{y}_{\ell_r})|_{x_1 = b_1, \ldots, x_n = b_n} \equiv 1\ (\text{mod }p).
$$
On the other hand, from Lemma~\ref{lem:primediv} we know that 
$$
C \cdot (y_{j_1} \cdots y_{j_k} \cdot \neg{y}_{\ell_1} \cdots \neg{y}_{\ell_r})|_{x_1 = b_1, \ldots, x_n = b_n} \equiv 0\ (\text{mod }p).
$$
Therefore, $C$ is divisible by $p$.

Summarizing everything, we get that for every prime $p \in \mathcal{P}$ either we have derived a Boolean equation $C' \cdot (y^2 - y)$ where $C'$ is divisible by $p$, or there is a clause $C \cdot y_{j_1} \cdots y_{j_k} \cdot \neg{y}_{\ell_1} \cdots \neg{y}_{\ell_r}$ where the constant $C$ is divisible by $p$.

Now, if the number of clauses in our CNF is at least $2^{n/3}$, then the first condition of the theorem holds. Suppose we have derived an unsatisfiable CNF with less then $2^{n/3}$ clauses. Then we have less than $2^{n/3}$ different variables in our CNF since it is unsatisfiable. Then we have derived less than $2^{n/3}$ equations of the form $C_i' \cdot (y_i^2 - y_i)$ and less than $2^{n/3}$ clauses of the form $C \cdot (y_{j_1} \cdots y_{j_k} \cdot \neg{y}_{\ell_1} \cdots \neg{y}_{\ell_r})$.

We showed that for any prime $p\in \mathcal{P}$ there is either an equation $C_i' \cdot (y_i^2 - y_i)$ such that $C_i'$ is divisible by $p$ or a clause $C \cdot (y_{j_1} \cdots y_{j_k} \cdot \neg{y}_{\ell_1} \cdots \neg{y}_{\ell_r})$ such that $C$ is divisible by $p$. So, since the total number of those equations is less then $2^{n/3 + 1}$, there is a constant $C$ (maybe $C=C_i'$) from one of those equations that is divisible by at least $\frac{|\mathcal{P}|}{2^{n/3 + 1}}$ prime numbers.

We know that the size of the set  $\mathcal{P}$ is at least $C'' \cdot 2^n / n$ by the Prime Number Theorem for some constant $C''$. Thus the constant $C$ should be divisible by at least $C''\cdot \frac{2^n}{2^{n/3 + 1} \cdot n}$ prime numbers, which is sufficient to satisfy the second or the third condition of the theorem.
\end{proof}

\begin{corollary}
Any {\extpcsqZ} derivation of an unsatisfiable CNF in $n$ variables from $\ebvp_n$ requires size $\Omega(2^{n/3})$.
\end{corollary}
\begin{proof}
If the number of clauses in this CNF is at least $2^{n/3}$, then our derivation already has size $\Omega(2^{n/3})$.

Otherwise, by Theorem~\ref{thm:lower_bound_integers} there is a constant $C$ in our derivation divisible by at least $\Omega\left(2^{n/3}\right)$ different prime numbers. Thus, the bit size of this integer should be $\Omega(2^{n/3})$.
\end{proof}

\subsection{Lower bound over the rationals}
In order to prove a lower bound over $\mathbb{Q}$, we need to convert an {\extpcsqQ} proof into an {\extpcsqZ} proof. We will use the following technical statement from \cite{alekseev2020lower}:

\begin{theorem}[\cite{alekseev2020lower}, Claim 12]\label{transaltion-from-z-to-q}
Suppose we have an ${\extpcsqQ}$ derivation $\{R_1, \ldots, R_t\}$ from some set of polynomials $\Gamma = \{f_1, \ldots, f_n\} \subset \mathbb{Z}[\bar x]$. Also, suppose $R_t \in \mathbb{Q}[\bar x]$, which means that $R_t$ does not depend on newly introduced variables.

Then there is an ${\extpcsqZ}$ derivation $\{R_1', \ldots, R_{t'}'\}$ from  $\Gamma$, where 
$$
R_{t'}' = \delta_1^{c_{1}} \cdots \delta_l^{c_{l}} \cdot  L_1^{c_{l + 1}} \cdots L_t^{c_{l + t}}  \cdot R_t
$$ and
\begin{itemize}
    \item $c_1, c_2, \ldots, c_{l + t}$ are some non-negative integers.  
    \item Each $L_i \in \mathbb{N}$ is the product of all denominators of coefficients of polynomial $R_i$.
    \item The set of constants $\{\delta_1, \delta_2, \ldots, \delta_l\} \subset \mathbb{N}$ is the set of all \textbf{denominators} of the constants in $\{\gamma_1, \gamma_2, \ldots, \gamma_l\}$, where  $\{\gamma_1, \gamma_2, \ldots, \gamma_l\} \subset \mathbb{Q}$ is the set of all constants $\alpha$ and $\beta$ occurring in linear combination steps in the proof. This means that some $R_j(\bar x, \bar y)$ was derived by using the linear combination rule with the constants $\alpha$ and $\beta$, or in other words, $R_j = \alpha R_i + \beta R_k$ for some previously derived polynomials $R_i$ and $R_k$.
\end{itemize}
\end{theorem}
\begin{note}
Observe that the size of the derivation $\{R_1', \ldots, R_{t'}'\}$ can be exponentially larger then the size of the derivation $\{R_1, \ldots, R_{t}\}$. However, this fact does not affect our proof, because in the next theorem we are concerned with divisibility only.
\end{note}

Now we will use \autoref{transaltion-from-z-to-q} to prove a lower bound over the rationals. 

\begin{theorem}\label{lower bound q CNF}
Any {\extpcsqQ} derivation of an unsatisfiable CNF from $\ebvp_n$  requires size $\Omega(2^{n/3})$.
\end{theorem} 
\begin{note}
Since division by integer numbers is allowed in {\extpcsqQ}, we can assume that the translation of the CNF has the following form: 
$$
p_i = y_{j_1} \cdots y_{j_k} \cdot \neg{y}_{\ell_1} \cdots \neg{y}_{\ell_r},
$$
and the translations of equations for Boolean variables has the form $y_i^2 - y_i = 0$.
\end{note}
\begin{proof}
If the number of clauses in this CNF is at least $2^{n/3}$, then our derivation already has size $\Omega(2^{n/3})$.

We can thus assume that the number of clauses is less than $2^{n/3}$. 

From \autoref{transaltion-from-z-to-q} we know that there is an {\extpcsqZ} derivation from $\ebvp_n$ where all the clauses have the following form:
$$
\delta_1^{c_{1}} \cdots \delta_l^{c_{l}} \cdot  L_1^{c_{l + 1}} \cdots L_t^{c_{l + t}}  \cdot y_{j_1} \cdots y_{j_k} \cdot \neg{y}_{\ell_1} \cdots \neg{y}_{\ell_r} = 0,
$$ 
and all the Boolean equations for the variables in those clauses also have the form
$$
\delta_1^{c_{1}} \cdots \delta_l^{c_{l}} \cdot  L_1^{c_{l + 1}} \cdots L_t^{c_{l + t}}  \cdot (y_i^2 - y_i) = 0.
$$
Then from  \autoref{thm:lower_bound_integers}  we know that for some clause or equation for Boolean variables $\delta_1^{c_{1}} \cdots \delta_l^{c_{l}} \cdot  L_1^{c_{l + 1}} \cdots L_t^{c_{l + t}}$ is divisible by at least $\Omega\left(2^{n/3}\right)$ different prime numbers.

Since $\delta_1, \ldots, \delta_l$, $L_1, \ldots, L_t$ are positive integers, we know that 
$
\delta_1 \cdots \delta_l \cdot  L_1 \cdots L_t
$
is divisible by at least $\Omega\left(2^{n/3}\right)$ different prime numbers. We also know that 
$$
\log \lceil \delta_1 \rceil + \cdots + \log \lceil \delta_l \rceil + \log \lceil L_1 \rceil  + \cdots +  \log \lceil L_t \rceil \le O(Size(S))
$$ 
because all constants $L_1, \ldots, L_t$ are products of denominators in the lines of our refutation $\{R_1, \ldots, R_t\}$ and all constants $\delta_1, \ldots, \delta_l$ are denominators of rationals in linear combinations used in our derivation. 

On the other hand, we know that for some constant $C''$ the following holds:
$$
\delta_1 \cdots \delta_l \cdot  L_1 \cdots L_t  \ge 2^{C'' \cdot 2^{n/3}}
$$
since our product is divisible by at least $\Omega\left( 2^{{n/3}} \right)$ different prime numbers. Therefore, $S \ge \Omega(2^{n/3})$.
\end{proof}

\section{Further research}\label{sec:oq}
A long-standing open question in semialgebraic proof complexity is to prove a superpolynomial lower bound for a rather week proof system (called {\LS} after Lov\'asz and Schrijver), namely for its most basic version \cite{Pud99}: consider only polynomials of degree at most two, express then as sums of monomials with coefficients written in binary, allow the addition and the multiplication by the input variable $x$ or its negation $1-x$ only. (That is, no arbitrary multiplication, no squares axioms ($f^2\ge0$), no extension variables.) Recently lower bounds on very strong proof systems have been proved for systems of polynomial equations (based on {\ebvp}) that do \emph{not} come from Boolean formulas. Does this generalization help to prove superpolynomial lower bounds for polynomial inequalities, for example, for {\LS}? 


We have shown a polynomial simulation of {\extps} proofs in {\extpc} augmented by the {\ebvp} rule, which was already known for stronger systems {\IPS} vs {\CPS} \cite{AGHT19_new}. How can we weaken the basic system so that the statement remains true? For example, following \cite{Bus87} we can simulate binary arithmetic in logarithmic depth (by formulas), which, unfortunately, gives only log${}^2n$ depth proofs. Is it possible to do better?


\section*{Acknowledgement}
We are grateful to Ilario Bonacina and Dima Grigoriev for fruitful discussions, and to Yuval Filmus for his detailed comments on an earlier draft of this paper.

\small
\bibliographystyle{alpha}

\begin{thebibliography}{AGHT20}

\bibitem[AGHT20]{AGHT19_new}
Yaroslav Alekseev, Dima Grigoriev, Edward~A. Hirsch, and Iddo Tzameret.
\newblock Semi-algebraic proofs, {IPS} lower bounds and the $\tau$-conjecture:
  Can a natural number be negative?
\newblock In {\em Proceedings of the 52nd Annual ACM Symposium on Theory of
  Computing (STOC 2020)}, pages 54--67, 2020.
\newblock Technical details can be found in ECCC TR19-142,
  https://eccc.weizmann.ac.il/report/2019/142.

\bibitem[Ale21]{alekseev2020lower}
Yaroslav Alekseev.
\newblock {A Lower Bound for Polynomial Calculus with Extension Rule}.
\newblock In Valentine Kabanets, editor, {\em 36th Computational Complexity
  Conference (CCC 2021)}, volume 200 of {\em Leibniz International Proceedings
  in Informatics (LIPIcs)}, pages 21:1--21:18, Dagstuhl, Germany, 2021. Schloss
  Dagstuhl -- Leibniz-Zentrum f{\"u}r Informatik.

\bibitem[Bus87]{Bus87}
Samuel~R. Buss.
\newblock Polynomial size proofs of the propositional pigeonhole principle.
\newblock {\em The Journal of Symbolic Logic}, 52(4), 1987.

\bibitem[CEI96]{CEI96}
Matthew Clegg, Jeffery Edmonds, and Russell Impagliazzo.
\newblock Using the {G}roebner basis algorithm to find proofs of
  unsatisfiability.
\newblock In {\em Proceedings of the 28th Annual ACM Symposium on the Theory of
  Computing (Philadelphia, PA, 1996)}, pages 174--183, New York, 1996. ACM.

\bibitem[CR79]{CR79}
Stephen~A. Cook and Robert~A. Reckhow.
\newblock The relative efficiency of propositional proof systems.
\newblock {\em J. Symb. Log.}, 44(1):36--50, 1979.

\bibitem[GH03]{GH03}
Dima Grigoriev and Edward~A. Hirsch.
\newblock Algebraic proof systems over formulas.
\newblock {\em Theoret. Comput. Sci.}, 303(1):83--102, 2003.
\newblock Logic and complexity in computer science (Cr\'eteil, 2001).

\bibitem[GHP02]{GHP02}
Dima Grigoriev, Edward~A. Hirsch, and Dmitrii~V. Pasechnik.
\newblock Complexity of semialgebraic proofs.
\newblock {\em Mosc. Math. J.}, 2(4):647--679, 805, 2002.

\bibitem[GP18]{GP14}
Joshua~A. Grochow and Toniann Pitassi.
\newblock Circuit complexity, proof complexity, and polynomial identity
  testing: The ideal proof system.
\newblock {\em J. {ACM}}, 65(6):37:1--37:59, 2018.

\bibitem[PT21]{PT_bvp_journal}
Fedor Part and Iddo Tzameret.
\newblock Resolution with counting: Dag-like lower bounds and different moduli.
\newblock {\em Comput. Complex.}, 30(1):2, 2021.

\bibitem[PTT21]{PTT21}
Fedor Part, Neil Thapen, and Iddo Tzameret.
\newblock First-order reasoning and efficient semi-algebraic proofs.
\newblock In {\em 36th Annual {ACM/IEEE} Symposium on Logic in Computer
  Science, {LICS} 2021, Rome, Italy, June 29 - July 2, 2021}, pages 1--13.
  {IEEE}, 2021.

\bibitem[Pud99]{Pud99}
Pavel Pudl{\'a}k.
\newblock On the complexity of the propositional calculus.
\newblock In {\em Sets and proofs ({L}eeds, 1997)}, volume 258 of {\em London
  Math. Soc. Lecture Note Ser.}, pages 197--218. Cambridge Univ. Press,
  Cambridge, 1999.

\bibitem[Tse68]{Tse68}
Grigori Tseitin.
\newblock {\em On the complexity of derivations in propositional calculus}.
\newblock Studies in constructive mathematics and mathematical logic Part II.
  Consultants Bureau, New-York-London, 1968.

\end{thebibliography}

\normalsize

\end{document}